\documentclass[conference]{IEEEtran}
\usepackage{graphicx,cite}
\usepackage{amsmath,amssymb}

\def\bx{\mathbf x}
\def\bn{\mathbf n}
\def\by{\mathbf y}

\def\bG{\mathbf G}
\def\bH{\mathbf H}
\def\bI{\mathbf I}
\def\bY{\mathbf Y}
\def\adag{\alpha^\dag}
\def\rhomax{\rho_\text{max}}
\def\maximize{\mathop{\text{maximize}}}
\def\leref#1{Lemma~\ref{#1}}
\def\figref#1{Figure~\ref{#1}}
\newtheorem{lemma}{Lemma}
\newtheorem{theorem}{Theorem}

\title{Joint Optimization of Power Allocation  and Training Duration for
Uplink Multiuser\\ MIMO Communications}
\author{
\authorblockN{Songtao Lu and Zhengdao Wang}
\authorblockA{Department of Electrical and Computer Engineering,
    Iowa State University, Ames, IA 50011, USA\\
    Emails: \{songtao, zhengdao\}@iastate.edu}}

\begin{document}
\maketitle
\begin{abstract}
In this paper, we consider a multiuser multiple-input multiple-output
(MU-MIMO) communication system between a base station equipped with multiple
antennas and multiple mobile users each equipped with a single antenna. The
uplink scenario is considered. The uplink channels are acquired by the base
station through a training phase. Two linear processing schemes are
considered, namely maximum-ratio combining (MRC) and zero-forcing (ZF). We
optimize the training period and optimal training energy under the average and
peak power constraint so that an achievable sum rate is maximized.
\end{abstract}


\section{Introduction}
\label{sec:intro}

Multiuser Multiple-input multiple-output (MU-MIMO) systems are a type of
cellular communication where the base station is equipped with multiple
antennas. The base station serves multiple mobile stations that are usually
equipped with a small number of antennas, typically one. MU-MIMO holds good
potentials for improving future communication system performance. However,
inter-cell interference or intra-cell interference decreases the achievable
performance of adopting MU-MIMO such that complex techniques are involved
inevitably. Therefore, there is a tradeoff of allocating the resources for
MU-MIMO, such as the number of users which has been studied in \cite{miyo13},
the training and feedback optimization in terms of training period and power
is considered for downlink MIMO broadcast channel in \cite{koji11}.

Recently, massive MIMO technology has attracted much attention. In such a
system, a large number of antennas at the base station are employed to improve
the system performance. There is already a body of results in the literature
about the analysis and design of large MIMO systems; see e.g., the overview
article \cite{rpll13} and references therein. It is important to quantify the
achievable performance of such systems in realistic scenarios. For example,
channel state information (CSI) acquisition in uplink takes time, energy, and
channel estimation error will always exist. For MU-MIMO, both circuit and
transmission powers in uplink were considered when designing power allocation
schemes \cite{miao13}, where energy efficiency is optimized.

In this paper, we are interested in performance of the uplink transmission in
a single-cell system. In particular, we ask what rates can be achieved in the
uplink by the mobile users if we assume realistic channel estimation at the
base station. The achievable rates of uplink MU-MIMO system with maximum-ratio
combining (MRC) and zero-forcing (ZF) detection were derived in \cite{nglm13},
and the performance evaluation was discussed in \cite{homa13}. But the
analysis therein assumes equal power transmission during the channel training
phase and the data transmission phase, which is not optimal in the sense that
sum rate is optimized. In previous study, the power allocation between
training phase and data phase was investigated for MIMO case with optimizing
effective signal-to-noise (SNR) in \cite{haho03}, which shows the tradeoff of
energy splitting between the two phases such that the achievable rate can be
maximized. However, the peak power was not considered before. If the training period is limited and the accurate estimate is required, then the peak power will be very high since we need to spend enough energy on training phase. In this case, the optimal power allocation strategy is not practical.

In this paper, the power allocation and training duration are both optimized
for uplink MU-MIMO systems in a systematic way. Two linear receivers, MRC and
ZF, are adopted with imperfect CSI. The average and peak power constraints are
both incorporated. We analyze the convexity of this optimization problem, and
derive the optimal solution. The solution is in closed form except in one case
where a one-dimensional search of a quasi-concave function is needed.
Simulation results are also provided to demonstrate the benefit of optimized
training, compared to equal power allocation considered in the literature.

The main contribution of the work is that we provided a complete solution for
the optimal training duration and training energy in an uplink MU-MIMO system
with either MRC (or ZF) receiver, and with both peak and average power
constraints.

\section{MU MIMO System Model}
\label{sec:system}
Consider an uplink MU-MIMO system. The base station is equipped with an array
of $M$ antennas. There are $K$ mobile users, each with a single antenna. We
assume $M>K$. For a massive MU-MIMO system, we have typically $M\gg K$.

\subsection{Transmission Scheme}

We assume that the channel follows a block fading model such that it remains
constant during a block of $T$ symbols, and changes independently from block
to block. The uplink transmission consists of two phases: training phase and
data transmission phase. The training phase lasts for $T_\tau$ symbols, and
the data phase lasts for $T_d=T-T_\tau$ symbols. We assume that the mobile
users are synchronized.

\subsubsection{Training}

The training phase is used for the base station to acquire the CSI. During the
training phase, the $K$ users send time-orthogonal signals at power level
$\rho_\tau$ per user. The training signals can be represented as a $K\times
T_{\tau}$ matrix $\sqrt{T_{\tau}\rho_{\tau}}\mathbf{\Phi}$, where $\mathbf
\Phi$ satisfies $\mathbf{\Phi}\mathbf{\Phi}^H=\bI_K$. The received signal in
training phase is
\begin{equation}
\bY_{\tau}=\sqrt{T_{\tau}\rho_{\tau}}\bH\mathbf{\Phi}+\mathbf{N}
\end{equation}
where $\bH$ denotes the fast fading channel matrix, whose entries are i.i.d.\
random variables which follow Gaussian distribution with zero mean and unit
variance, i.e., $\mathcal{CN}(0,1)$; $\mathbf{N}$ is an $M\times T_{\tau}$
matrix with i.i.d. $\mathcal{CN}(0,1)$ elements that represent the additive
noise. The minimum mean-square error (MMSE) estimate
\begin{equation}
\widehat\bH = \frac{\sqrt{T_\tau\rho_\tau}}{T_\tau\rho_\tau+1} \bY_\tau
\mathbf{\Phi}^H
\end{equation}
will be used for demodulating the data symbols during the data transmission
phase. Note that we require $T_\tau\ge K$ to satisfy the time-orthogonality.

\subsubsection{Data Transmission}

In the data transmission phase, all users send signals with power $\rho_d$.
The received signal is
\begin{equation}\label{eq.received}
\by=\sqrt{\rho_d}\bH\bx+\bn
\end{equation}
where $\bx$ is a $K\times 1$ vector denoting the transmitted symbols, and
$\bn$ is a $M\times 1$ vector denoting the additive noise. We assume that the
noise distribution is i.i.d.\ $\mathcal{CN}(0,1)$.

\subsection{Linear Receivers}

Based on the model in \eqref{eq.received}, we will consider two linear
demodulation schemes: MRC and ZF receivers. The demodulated symbols can be
written as
\begin{equation}
\widehat{\bx}=\bG\by
\end{equation}
where $\bG$ is a $K\times M$ matrix that depends on the receiver type.

\subsubsection{MRC Receiver} For MRC receiver, we have $\bG=\widehat \bH^H$.
The signal-to-inference plus noise ratio (SINR) for any of the $K$ users'
symbols can be obtained in the same way as in, e.g., \cite[eq.~(39)]{nglm13}
as
\begin{equation}\label{eq.mrcsinr}
\textsf{SINR}^{\text{MRC}}=
  \frac{T_{\tau}\rho_{\tau}\rho_d(M-1)}
    {T_{\tau}\rho_{\tau}\rho_d(K-1)+K\rho_d+T_{\tau}\rho_{\tau}+1}.
\end{equation}

\subsubsection{ZF Receiver} For ZF receiver, $\bG=\frac 1{\sqrt{\rho_d}}
(\widehat{\bH}^{H}\widehat{\bH})^{-1}\widehat{\bH}^H$. The expected value of
SINR for any of the $K$ users' symbols can be obtained in the same way as in,
e.g., \cite[eq.~(42)]{nglm13} as
\begin{equation}\label{eq.zfrate}
\textsf{SINR}^{\text{ZF}}=
  \frac{T_{\tau}\rho_{\tau}\rho_d(M-K)}{K\rho_d+T_{\tau}\rho_{\tau}+1}.
\end{equation}
For either receiver, a lower bound on the sum rate achieved by the $K$ users
is given by
\begin{equation}\label{eq.R}
R^{\mathcal{A}}(\alpha,T_d)=
  \frac{T_d}{T}K\log_2(1+\textsf{SINR}^{\mathcal{A}})
\end{equation}
where $\mathcal{A}\in \{\text{MRC}, \text{ZF}\}$.

\subsection{Power Allocation} We assume that the transmitters are subject to
both peak and average power constraints.

\subsubsection{Average Power Constraint} We assume the average transmitted
power over one coherence interval $T$ is equal to a given constant $\rho$,
namely
\(
\rho_dT_d+\rho_{\tau}T_{\tau}=\rho T \label{eq.ener}
\).
Let $\alpha:= \rho_\tau T_\tau/(\rho T)$ denote the fraction of the total
transmit energy that is devoted to channel training; i.e.,
\begin{equation} \label{eq.optfra}
\rho_{\tau}T_{\tau}=\alpha\rho T, \quad \rho_dT_d=(1-\alpha)\rho T,
\quad 0\le \alpha\le 1.
\end{equation}

\subsubsection{Peak Power Constraint}

The peak power during the transmission is assumed to be no more than
$\rhomax$; i.e.,
\begin{equation}\label{eq.oricon1}
 0\le\rho_d, \rho_{\tau}\le \rhomax.
\end{equation}

\subsection{Optimization Problem}
\def\objfunc#1{R^{#1}(\alpha,T_d)}
For an adopted receiver, $\mathcal{A}\in \{\text{MRC}, \text{ZF}\}$, our goal
is to maximize the uplink achievable rate subject to the peak and average
power constraints. That is,
\begin{align}\label{eq.oripro}
\maximize_{\alpha,T_d}\quad&\objfunc{\mathcal A} \\
\text{subject to}\quad& T_d+T_{\tau}=T \\
& \rho T\alpha+\rhomax T_d\le\rhomax T \label{eq.alcon1} \\
&  -\rho T\alpha-\rhomax T_d\le-\rho T \label{eq.alcon2}  \\
& \quad 0\le \alpha\le 1 \label{eq.alconalpha} \\
& 0<T_d\le T-K  \label{eq.tdconstr}
\end{align}
where $\objfunc{\mathcal A}$ is as given in \eqref{eq.R}; \eqref{eq.alcon1}
and \eqref{eq.alcon2} are from the peak power constraints in the training and
data phases, respectively; and the last constraint is from the requirement
that $T_\tau \ge K$.

\section{SINR Maximization with $\alpha$ for Fixed $T_d$}

The feasible set of the optimization problem
\eqref{eq.oripro}--\eqref{eq.tdconstr} is convex, but the convexity of the
objective function is not obvious. In this section, we consider the
optimization problem when $T_d$ is fixed. In this case, we will prove that
$R^{\mathcal{A}}(\alpha,T_d)$ is concave in $\alpha$, and derive the optimized
$\alpha$. The result will be useful in the next section that $\alpha$ and
$T_d$ are jointly optimized.

For a fixed $T_d$, from the peak power constraints \eqref{eq.alcon1} and
\eqref{eq.alcon2}, we have
\begin{equation}\label{eq.solcon}
\frac{\rhomax T_{\tau}}{\rho T}+\left(1-\frac{\rhomax}{\rho}\right)\le\alpha\le\frac{\rhomax T_{\tau}}{\rho T}.
\end{equation}
Combined with \eqref{eq.alconalpha}, the overall constraints on $\alpha$ is
\begin{equation}\label{eq.solcon2}
\min\{0,\frac{\rhomax T_{\tau}}{\rho T}+\left(1-\frac{\rhomax}{\rho}\right)\}\le\alpha\le\max\{\frac{\rhomax T_{\tau}}{\rho T},1\}.
\end{equation}
In the remaining part of this section, we will first ignore the peak power
constraint, and derive the optimal $\alpha\in(0,1)$ for a given $T_d$. At the
end of this section, we will reconsider the effect of the peak power
constraint on the optimal $\alpha$.

\subsection{MRC Case without peak power constraint}\label{sec.mrc}

Using \eqref{eq.optfra} we can rewrite \eqref{eq.mrcsinr} as
\begin{equation}\label{eq.newmrcsinr}
\textsf{SINR}^\text{MRC}(\alpha)=\frac{M-1}{K-1}
  \frac{\alpha(\alpha-1)}{\alpha^2-a_1\alpha-b_1}
\end{equation}
where
\begin{equation}\label{eq.defab}
a_1=1+\frac{T_d-K}{\rho T(K-1)},\quad b_1=\frac{\rho T K+T_d}{\rho^2T^2(K-1)}>0.
\end{equation}
It can be verified that $1-a_1-b_1\le0$.

\subsubsection{Behavior of the $\textsf{SINR}^\text{MRC}(\alpha)$ function}

Define
\begin{equation}
g(\alpha):=\textsf{SINR}^\text{MRC}\cdot(K-1)/(M-1).
\end{equation}
And let $g_d(\alpha)=\alpha^2-a_1 \alpha -b_1$, which is the denominator of
$g(\alpha)$.

\emph{Remark 1}: It can be observed that when $1-a_1-b_1\le0$ and $b_1>0$,
$g_d(\alpha)$ is negative at both $\alpha=0$ and $\alpha=1$. Since the leading
coefficient of $g_d(\alpha)$ is positive, $g_d(\alpha)<0$ for
$\alpha\in(0,1)$, and it has no root in $(0,1)$.

\begin{lemma}\label{le.1}
The function $g(\alpha)$ is concave in $\alpha$ over $(0,1)$ when
$1-a_1-b_1\le0$ and $b_1>0$.
\end{lemma}
\begin{proof}
See Appendix A.
\end{proof}
According to \leref{le.1}, we know that there is a global maximal point for
\eqref{eq.newmrcsinr}. Take the derivative of \eqref{eq.newmrcsinr} and set it
as 0, we have
\begin{equation}\label{eq.equation}
(1-a_1)\alpha^2-2b_1\alpha+b_1=0.
\end{equation}

Based on Remark 1, we deduce that $g(\alpha)>0$ for $\alpha\in (0, 1)$. In
addition, we have $g(0)=0$ and $g(1)=0$. Therefore, there is an optimal
$\alpha$ within $(0,1)$ rather than at boundaries.

\subsubsection{The optimizing $\alpha$} We discuss the optimal $\alpha$ in
three cases, depending on $T_d$, as compared to $K$.
\begin{enumerate}
\item If $T_d=K$, then $1-a_1=0$. Hence, we have $\alpha^*=1/2$, and
\begin{equation}
\textsf{SINR}^{\text{MRC}}(\frac{1}{2})=\frac{M-1}{K-1}\frac{1/4}{1/4+\frac{K(\rho T+1)}{\rho^2T^2(K-1)}}
\end{equation}

\item If $T_d<K$, then $1-a_1>0$. Since $b_1>1-a_1$, $b_1/(1-a_1)>1$. Between
the two roots of \eqref{eq.equation}, the one in between 0 and 1 is
\begin{equation}\label{eq.alphale}
\alpha^*=\frac{b_1-\sqrt{b_1(a_1+b_1-1)}}{1-a_1}.
\end{equation}

\item If $T_d>K$, then $1-a_1<0$. It can be deduced that in this case
$\alpha^*$ in \eqref{eq.alphale} is still between 0 and 1 and therefore is the
optimal $\alpha$.
\end{enumerate}

Substituting \eqref{eq.defab} into \eqref{eq.alphale}, we have
\begin{equation}\label{eq.astar}
\alpha^*=\frac{\sqrt{(\rho TK+T_d)(\rho TT_d+T_d)}-(\rho TK+T_d)}{\rho T(T_d-K)}.
\end{equation}
We can simplify the expression for the optimal $\alpha$ at high and low SNR:
\begin{enumerate}
\item At high SNR, the optimal $\alpha^*$ is
\begin{equation}
\alpha^*_{\rm H}\approx \frac{\sqrt{KT_d}-K}{T_d-K}.
\end{equation}

\item Similarly, at low SNR, the optimal $\alpha^*$ is
\begin{align}
\alpha^*_{\rm L}\approx\frac{1}{2}.
\end{align}
As a result, $\textsf{SINR}^{\text{MRC}}(\alpha^*_{\rm L})=(M-1)/(4T_d(K-1))$.
If the SNR is low, the fraction between the training and data is independent
on the system parameters $M$, $K$, $\rho_d$, $\rho_{\tau}$, $T_\tau$, and $T$.
\end{enumerate}

\subsection{ZF Case without peak power constraint}

This optimization problem in the ZF case is similar to that in \cite[eq.
22]{haho03}, which maximizes the effective SNR for MIMO system with MMSE
receiver. Here, we only give the final optimization results.

\subsubsection{The SINR function}

Using \eqref{eq.optfra} we can rewrite \eqref{eq.zfrate} as
\begin{equation} \label{eq.zfsinr}
\textsf{SINR}^{\text{ZF}}(\alpha)=\frac{T\rho(M-K)\alpha(1-\alpha)}{(T_d-K)(\gamma+\alpha)}
\end{equation}
where $\gamma=\frac{K\rho T+T_d}{\rho T(T_d-K)}$. The second derivative of
\eqref{eq.zfsinr} is
\begin{equation}
-2T\rho(M-K)\gamma(\gamma+1)/[(T_d-K)(\alpha+\gamma)^3],
\end{equation}
which we will consider to decide the convexity of the objective function. It
can be verified that in all the three cases, namely $T_d=K$, $T_d>K$, and
$T_d<K$, $\textsf{SINR}^{\text{ZF}}(\alpha)$ is concave in $\alpha$ within
$\alpha\in(0,1)$.

\subsubsection{The optimizing $\alpha^*$}

Taking the first derivative of \eqref{eq.zfsinr} and set to 0, we can obtain
the optimal $\alpha^*$:
\begin{enumerate}
\item When $T_d=K$, $\alpha^*=1/2$.

\item When $T_d>K$, $\alpha^*=-\gamma+\sqrt{\gamma(\gamma+1)}$.

\item When $T_d<K$, $\alpha^*=-\gamma-\sqrt{\gamma(\gamma+1)}$.
\end{enumerate}
We can simplify the expression for the optimal $\alpha$ at high and low SNR:
\begin{enumerate}
\item At high SNR, $\gamma=K/(T_d-K)$.

\item At low SNR, $\gamma=T_d/(\rho T(T_d-K))$, $\alpha^*_{\rm L}=1/2$ which
is consistent with the MRC case.
\end{enumerate}

\subsection{MRC and ZF with peak power constraint}

So far we have ignored the peak power constraint. When the peak power is
considered, and $\alpha^*$ is not within the feasible set \eqref{eq.solcon2},
the optimal $\widetilde{\alpha}^*$ with the peak power constraint is the
$\alpha$ within the feasible set that is closest to the $\alpha^*$ we derived,
which is at one of the two boundaries of the feasible set, due to the
concavity of the objective function.

\section{Achievable Rate Maximization with $\alpha$ and $T_d$}

In this section, $\alpha$ and $T_d$ are jointly optimized for maximizing the
achievable rate of uplink MU-MIMO system as illustrated in
\eqref{eq.oripro}--\eqref{eq.tdconstr} when both average and peak power
constraints are considered.

The feasible set with respective to $\alpha$ and $T_d$ is illustrated in
Fig.~1. It can be observed that the feasible region is in between the
following two lines
\begin{eqnarray}
T_d&=&-\rho T\alpha/\rhomax+T, \label{eq.slab1} \\ T_d&=&-\rho
T\alpha/\rhomax+\rho T/\rhomax \label{eq.slab2}
\end{eqnarray}
where $\alpha$ and $T_d$ satisfy \eqref{eq.alconalpha} and
\eqref{eq.tdconstr}.
\begin{figure}[htp]
\centering \includegraphics[width=\linewidth]{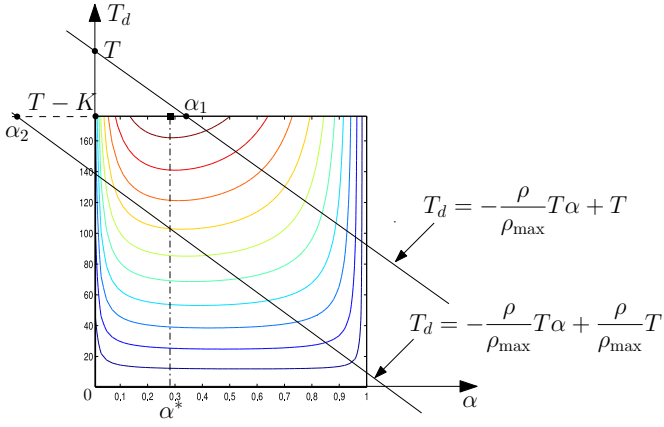}
\caption{Feasible region and the contour of the objective function in the MRC
case; $T=196$, $K=20$ and $M=50$.}
\label{fig.feas}
\end{figure}

We have the following lemma that is useful for describing the behavior of our
objective function $\textsf{SINR}^{\mathcal A}(\alpha, T_d)$ when $\alpha$ is
fixed.

\begin{lemma} \label{lemma.2}
The function $f(x)=x\ln(1+a/(b+cx))$, when $a,b,c,x>0$, is concave and
monotonically increasing.
\end{lemma}
\begin{proof}
See Appendix B.
\end{proof}

In summary, the convexities of the objective function are known to have the
following two properties:
\begin{enumerate}
\item [(\textbf{P1})] From \leref{le.1}, for fixed $T_d$, $R^{\mathcal{A}}$ is
a concave function with respect to $\alpha$.

\item [(\textbf{P2})] From \leref{lemma.2}, for fixed $\alpha$,
$R^{\mathcal{A}}$ is a concave function and monotonically increasing with
respect to $T_d$.
\end{enumerate}
Since the feasible set is convex, our optimization problem
\eqref{eq.oripro}--\eqref{eq.tdconstr} is a biconvex problem that may include
multiple local optimal solutions. However, after studying the convexities of
the objective function, there are only three possible cases for the optimal
solutions, as we discuss below.

In the remainder of this section, let $\adag$ denote the optimal $\alpha$ when
$T_d=T-K$.

\subsection{Case 1: $\rho_{\tau}$ is limited by $\rhomax$}

Define $\alpha_1:=\rhomax K/\rho T$, which is the root of $T-K=-\rho
T\alpha/\rhomax+T$ in $\alpha$; see \figref{fig.feas}. In the case where
$\alpha_1<\adag$, because of the property P2 the optimal $(\alpha^*, T_d^*)$
must be on one of the two lines given by i) $T_d=-\rho T\alpha/\rhomax+T$,
$\alpha\in [\alpha_1, 1]$, and ii) $T_d=T-K$, $\alpha \in [0, \alpha_1]$.

On the line $T_d=T-K, \alpha\in[0,\alpha_1]$ the objective function is concave
and increasing with $\alpha$, thanks to property P1. Hence, we only need to
consider the line $T_d=-\rho T\alpha/\rhomax+T, \alpha\in[\alpha_1,1]$.

\begin{lemma}\label{le.quasi}
The objective function $\objfunc{\text{MRC}}$ along the line $T_d=-\rho
T\alpha/\rhomax+T, \alpha\in[\alpha_1,1]$ is quasiconcave in $\alpha$.
\end{lemma}
\begin{proof}
Consider MRC processing. Substituting \eqref{eq.slab1} into
$\objfunc{\text{MRC}}$, we have
\begin{equation}\label{eq.conmrcrate}
R^{\text{MRC}}(\alpha)=
  \frac{K}{T}\left(-\frac{\rho T}{\rhomax}\alpha+T\right)
  \log_2(1+\textsf{SINR}^{\text{MRC}})
\end{equation}
where
\begin{equation}\label{eq.sinrmrcq}
\textsf{SINR}^{\text{MRC}}(\alpha)=
  \frac{\alpha(\alpha-1)\rho^2T^2(M-1)}{a_2\alpha^2-b_2\alpha-c_2},
\end{equation}
and $a_2=\rho^2T^2(K-1)+\rho^2T^2/\rhomax$, $b_2=\rho^2T^2(K-1)+\rho T^2-\rho
TK-\rho T\alpha/\rhomax$ and $c_2=K\rho T+T$. Since
$R^{\text{MRC}}(\alpha)>0$, in order to prove the quasi-concavity of
$R^{\text{MRC}}(\alpha)$, we need to prove that the super-level set
$\mathcal{S}_{\beta}=\{\alpha|0<\alpha<1, R^{\text{MRC}}(\alpha)\ge\beta\}$
for each $\beta\in\mathbb{R}^{+}$ is convex. Equivalently, if we define
\begin{equation}\label{eq.quasiproof}
\phi_{\beta}(\alpha)=
  \frac{\beta}{\frac{K}{T}(\frac{\rho T\alpha}{\rhomax}-T)}
    +\log_2(1+\textsf{SINR}^{\text{MRC}}(\alpha)).
\end{equation}
we only need to prove that
$\mathcal{S}_{\phi}=\{\alpha|0<\alpha<1,\phi_{\beta}(\alpha)\ge0\}$ is a
convex set.

It can be checked that the first part of $\phi_\beta(\alpha)$, namely
$\beta/[{\frac{K}{T}(\frac{\rho T\alpha}{\rhomax}-T)}]$, is concave for
$\alpha\in [0, 1]$. For the other part of $\phi_{\beta}(\alpha)$, from
\eqref{eq.sinrmrcq} we know that
\begin{equation}
a_2-b_2-c_2=\rho T(\frac{\rho}{\rhomax}-1)-T(1-\alpha\frac{\rho}{\rhomax})<0
\end{equation}
where $a_2,c_2>0$. Applying \leref{le.1}, we know
$\textsf{SINR}^{\text{MRC}}(\alpha)$ is concave. Hence,
$\log_2(1+\textsf{SINR}^{\text{MRC}}(\alpha))$ is also concave since function
$\log(1+x)$ is concave and nondecreasing \cite{bova11}. Therefore, its
super-level set $\mathcal{S}_{\phi}$ is convex. It follows that the
super-level set $\mathcal{S}_{\beta}$ of $R^{\text{MRC}}(\alpha)$ is convex
for each $\beta\ge 0$. The objective function is thus quasiconcave.
\end{proof}

Thanks to \leref{le.quasi}, we can find the optimal $\alpha$ by setting the
derivative of \eqref{eq.conmrcrate} with respect to $\alpha$ to 0. Efficient
one-dimensional searching algorithm such as Newton method or bisection
algorithm \cite{bova11}, can be adopted to find out the optimal $\alpha$.

\subsection{Case 2: $\rho_d$ is limited by $\rhomax$}

Define $\alpha_2:=1-\rhomax(T-K)/\rho T$, which is the root of $T-K=\rho
T\alpha/\rhomax+\rho T/\rhomax$ in $\alpha$. If $\alpha_2>\adag$, because of
the property P2 the optimal $(\alpha^*, T_d^*)$ must be on one of the two
lines given by i) $T_d=-\rho T\alpha/\rhomax+T, \alpha\in(\alpha_1,1)$,
$\alpha\in [\alpha_1, 1]$, and ii) $T_d=T-K$, $\alpha \in [\alpha_2,
\alpha_1]$. Along the line $T_d=T-K, \alpha\in(\alpha_1,1)$, the corresponding
function is decreasing in $\alpha$ because of the property P1. Also
considering P2, which implies that the optimal point in this case cannot
include $T_d<T-K$, we conclude that the point $(\alpha^*, T_d^*)=(\alpha_2,
T-K)$ is the global optimal solution of the problem.

\subsection{Case 3: Neither $\rho_d$ nor $\rho_{\tau}$ is not limited by
$\rhomax$}

If $\alpha_2<\adag<\alpha_1$, the optimal point is achieved at $(\alpha^*,
T_d^*)=(\adag, T-K)$, according to properties P1 and P2.

Summarizing what we have discussed so far, we have the following theorem.
\begin{theorem}
For the MRC receiver, set $\adag=1/2$ if $T_d=K$ and otherwise set $\adag$
according to \eqref{eq.astar} when $T_d=T-K$. Set $\alpha_1=\rhomax K/\rho T$
and set $\alpha_2= 1-\rhomax (T-K)/\rho T$. There are three cases: Case 1) If
$\alpha_1<\adag$, then $\alpha^*$ is given by the maximizer of
$R^{\text{MRC}}(\alpha)$ in \eqref{eq.conmrcrate}, and $T_d^*=-\rho T\alpha^*
/\rhomax + T$; Case 2) If $\alpha_2>\adag$ then $(\alpha^*, T_d^*)=(\alpha_2,
T-K)$; Case 3) If $\alpha_2<\adag<\alpha_1$, then $(\alpha^*, T_d^*)=(\adag,
T-K)$.
\end{theorem}

We also have similar results regarding the optimal energy allocation factor
$\alpha$ and training period $T_\tau$ for the ZF case. Due to the space limit,
the result is not included here.

We also remark that our results are applicable for any $M>K$. When $M\gg K$,
the system is known as a ``massive MIMO'' system. Our results offer optimal
training energy allocation and optimal training duration when there is a peak
power constraint in addition to the average power constraint.

\section{Numerical Results} \label{sec:simulations}

In this section, we compare the achievable rates between equal power
allocation scheme and our optimized one under average and peak power
constraints. In our simulations, we set $\rho_{\max}=1.2\rho$, $K=10$, and
$T=196$. We consider the following schemes: 1) MRC, which refers to the case
where MRC receiver is used and the same average power is used in both training
and data transmission phases \cite{nglm13}. 2) optimized MRC, which refers to the case where
MRC receiver is used, the training duration is $K$, and there is no peak power
constraint. 3) power-limited MRC, where MRC receiver is used, and both the
training duration and training energy are optimized under both the average and
peak power constraints. We will also consider the ZF variants of the above
three cases, namely ZF, optimized ZF, and power-limited ZF. The energy
efficiency is defined as $\eta^{\mathcal A}:=\objfunc{\mathcal A}/\rho$.

In Fig.~2, we show the achieved rates of various schemes as the number of
antennas increases. It can be seen that the optimized MRC (ZF) performs better
than the unoptimized MRC (ZF) as well as the peak-power limited MRC (ZF). In
Fig.~3, the energy efficiency is shown as a function of $\rho$. It can be seen
that there is an optimal average transmitted power for maximal energy
efficiency. It can also be seen that optimized schemes show a significant gain
when $\rho$ is small, since the power resource is precious. In Fig.~4, we show
the energy efficiency versus sum rate. It can be observed that the energy
efficiency is maximized at a certain rate. In particular, the optimized
schemes achieve higher energy efficiencies. Also from the simulations, we can
see that ZF performs better than MRC at high SNR, but worse when SNR is low.

\begin{figure}[htp]
\begin{minipage}[b]{1.0\linewidth}
  \centering
  \centerline{\includegraphics[width=7.6cm]{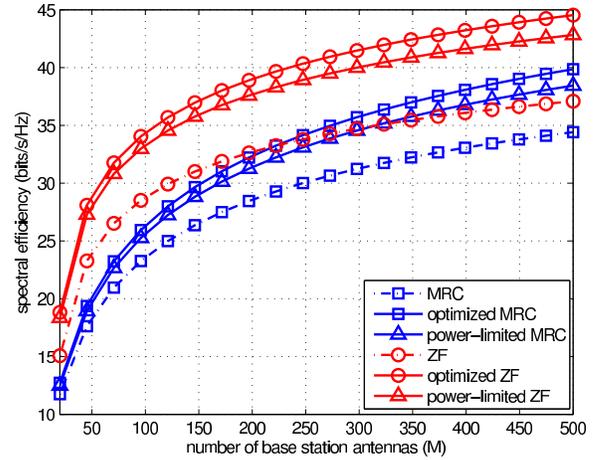}}
  \caption{Comparison between equally power allocation and the optimized fraction $\alpha$ in terms of the number of base station antennas, where
  $\rho=-5$dB.}\medskip\label{fig.image1}
\end{minipage}
\end{figure}
\begin{figure}[htp]
\begin{minipage}[b]{1.0\linewidth}
  \centering
  \centerline{\includegraphics[width=7.6cm]{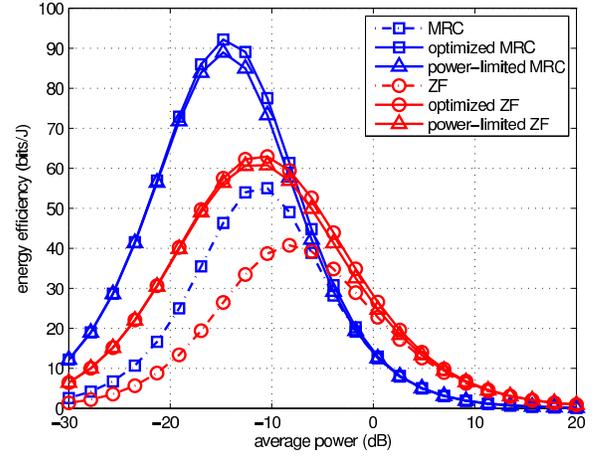}}
  \caption{Comparison of energy efficiency in terms of the transmitted power $\rho$, where $M=20$.}\medskip
\end{minipage}
\end{figure}
\begin{figure}[htp]
\begin{minipage}[b]{1.0\linewidth}
  \centering
  \centerline{\includegraphics[width=7.6cm]{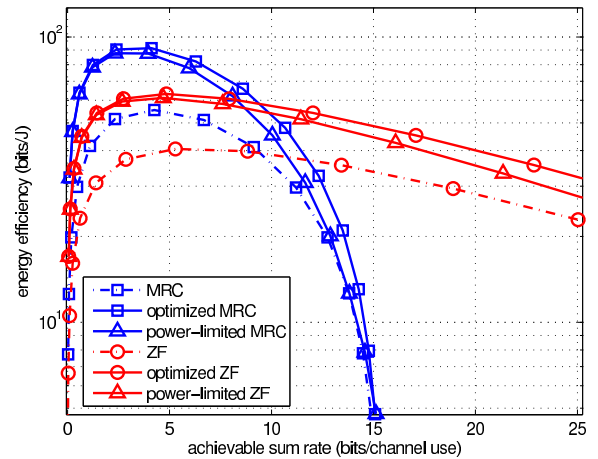}}
  \caption{Comparison of energy efficiency versus the spectral efficiency, where $M=20$.}\medskip
\end{minipage}
\end{figure}
\section{Conclusion}
\label{sec:conclusion}
In this paper, we considered a uplink MU-MIMO system with training and data
transmission phases. Two receivers were considered, namely MRC and ZF
receivers. Power allocation between training phase and data phase, and the
training duration were optimized such that an achievable rate was maximized.
Both average and peak power constraints were considered. We performed a
careful analysis of the convexity of the problem and derived optimal solution
either in closed form or in one case through a one-dimensional search for a
quasi-concave function. Our results were illustrated through numerical
examples, for some example system setups including those with a large number
of antennas at the base station.

\section{Appendix}

\subsection{Proof of Lemma 1} Replacing $\alpha$ as $x$ in
\eqref{eq.newmrcsinr}, we need to verify that the second derivative of
\eqref{eq.newmrcsinr} with respective to $x$ is negative \cite{bova11}. The
first derivative of $g(x)$ is
\begin{equation}
g'(x)=\frac{(1-a)x^2-2bx+b}{(x^2-ax-b)^2}
\end{equation}
where $1-a-b<0$, $b>0$ and $x\in(0,1)$. Then, take the second derivative of
$g(x)$, we have
\begin{equation}
g''(x)=\frac{2}{(x^2-ax-b)^3}\underbrace{\big((a-1)x^3+3bx^2-3bx+ab+b^2\big)}_{f(x)}
\end{equation}
From Remark 1 and $b>0$, we know that $(x^2-ax-b)^3<0$. The goal of the proof
becomes to show $f(x)>0$.

Checking the boundary of $f(x)$, we know that
\begin{equation}
f(0)=ab+b^2=b(b+a)>b>0,
\end{equation}
\begin{equation}
f(1)=ab+b^2=a-1+ab+b^2=(a+b-1)(b+1)>0.
\end{equation}
Next, we need to consider the monotonicity of the function during the interval
$x\in(0,1)$. Take the derivative of $f(x)$, we get
\begin{align}\label{eq.fxderiv}
f'(x)=&3(a-1)x^2+6bx-3b
\nonumber \\
&=3(a-1)(x^2+\frac{2b}{a-1}x-\frac{b}{a-1}),
\end{align}
which is a quadratic function.

When $a=1$, $f'(x)=6bx-3b=3b(2x-1)$. The function is decreasing until $x=1/2$
and increasing afterwards.
Since
\begin{equation}
f(\frac{1}{2})=\frac{1}{4}b+b^2>0,
\end{equation}
it can be deduced that $f(x)>0$.

When $a\ne 1$, we know that $f'(1)=3(a+b-1)>0$, $f'(0)=-3b$, meaning that the
function $f(x)$ is decreasing first and increasing after the minimum point.

Here, we need to verify the minimum value of $f(x^*)$ is always greater than
0. According to \eqref{eq.fxderiv}, the minimum point given by the root of
$f'(x^*)=0$ is
\begin{equation}\label{eq.minpoint}
x^*=-\frac{b}{a-1}+\sqrt{\frac{b(a+b-1)}{(a-1)^2}},
\end{equation}
since $a+b>1$ and $b>0$. Substituting \eqref{eq.minpoint} into $f(x)$, we have
\begin{align}
f(x)= & x((a-1)x^2+2bx-b)+bx^2-2bx+ab+b^2
\nonumber \\
\mathop{=}\limits^{(a)} & bx^2+\frac{2b^2}{a-1}x-\frac{b^2}{a-1}-\frac{2b^2}{a-1}x+\frac{b^2}{a-1}
\nonumber \\
&-2bx+ab+b^2
\nonumber \\
\mathop{=}\limits^{(a)} &-\frac{2b(b+a-1)}{a-1}x+\frac{ab(a+b-1)}{a-1}
\nonumber \\
= &\frac{b(a+b-1)}{a-1}\underbrace{(\frac{2b}{a-1}-2\sqrt{\frac{b(a+b-1)}{(a-1)^2}}+a)}_{h(x)}
\end{align}
where (a) is according to $f'(x^*)=0$.

For $a-1>0$,
\begin{align}
h(x)=&\frac{2b}{a-1}-\frac{\sqrt{b(a+b-1)}}{a-1}+a
\nonumber \\
=&\frac{2}{a-1}\frac{b^2-b(a+b-1)}{b+\sqrt{b(a+b-1)}}+a \\
\mathop{>}\limits^{(b)}&a-\frac{2b}{b+\sqrt{b^2}}>0
\end{align}
where (b) is based on $a-1>0$. Therefore, $f(x)>0$.

For $a-1<0$,
\begin{align}
h(x)&=\frac{2b}{a-1}+\frac{2\sqrt{b(a+b-1)}}{a-1}+a\\
&\mathop{<}^{(c)}\frac{2(1-a)}{a-1}+a<0,
\end{align}
where (c) is due to $b>1-a$. Hence, $f(x)>0$.

\subsection{Proof of Lemma 2} The derivative of $f(x)=x\ln(1+a/(b+cx))$, where
$a,b,c,x>0$, is
\begin{equation}
f'(x)=\ln(1+\frac{a}{cx+b})-\frac{acx}{(cx+a+b)(cx+b)}
\end{equation}
It is clear that $\lim_{x\to\infty}f'(x)=0$. If we can verify that the
function $f'(x)$ is monotonically decreasing, then $f'(x)$ is always positive.
Hence, we take the second derivative of $f'(x)$, and get
\begin{equation} \label{eq.secderivtd}
f''(x)=-\frac{abc^2x+ac^2(a+b)x+2ac(a+b)b}{[(cx+b)(cx+a+b)]^2}<0,
\end{equation}
since $a,b,c,x>0$. This means that $f'(x)$ is decreasing. Therefore, $f'(x)$
is always positive, i.e., $f(x)$ is an increasing and concave function.

\noindent
\emph{Acknowledgement:} The work in this paper was support in part by NSF
Grants No.~1218819 and No.~1308419.

\bibliographystyle{IEEE-unsorted}
\bibliography{refs}

\begin{thebibliography}{1}

\bibitem{miyo13}
M.~Jung, Y.~Kim, J.~Lee, and S.~Choi,
\newblock ``Optimal number of users in zero-forcing based multiuser mimo
  systems with large number of antennas,''
\newblock {\em J Commun. Netw.}, vol.~15, no.~4, pp.~362--369, Aug. 2013.

\bibitem{koji11}
M.~Kobayashi, N.~Jindal, and G.~Caire,
\newblock ``Training and feedback optimization for multiuser {MIMO} downlink,''
\newblock {\em {IEEE} Trans. Commun.}, vol.~59, no.~8, pp.~2228--2240, Aug.
  2011.

\bibitem{rpll13}
F.~Rusek, D.~Persson, B.~K. Lau, E.~G. Larsson, T.~L. Marzetta, O.~Edfors, and
  F.~Tufvesson,
\newblock ``Scaling up {MIMO:} opportunities and challenges with very large
  arrays,''
\newblock {\em IEEE Signal Process. Mag.}, vol.~30, no.~1, pp.~40--60, Jan.
  2013.

\bibitem{miao13}
G.~Miao,
\newblock ``Energy-efficient uplink multi-user {MIMO},''
\newblock {\em {IEEE} Trans. Wireless Commun.}, vol.~12, no.~5, pp.~2302--2313,
  May 2013.

\bibitem{nglm13}
H.~Q. Ngo, E.~G. Larsson, and T.~L. Marzetta,
\newblock ``Energy and spectral efficiency of very large multiuser {MIMO}
  systems,''
\newblock {\em {IEEE} Trans. Commun.}, vol.~61, no.~4, pp.~1436--1449, Apr.
  2013.

\bibitem{homa13}
H.~Yang and T.~Marzetta,
\newblock ``Performance of conjugate and zero-forcing beamforming in
  large-scale antenna systems,''
\newblock {\em {IEEE} J. Select. Areas Commun.}, vol.~31, no.~2, pp.~172--179,
  Feb. 2013.

\bibitem{haho03}
B.~Hassibi and B.~M. Hochwald,
\newblock ``How much training is needed in multiple-antenna wireless links?''
\newblock {\em {IEEE} Trans. Info. Theory}, vol.~49, no.~4, pp.~951--963, Apr.
  2003.

\bibitem{bova11}
S.~Boyd and L.~Vandenberghe,
\newblock {\em Convex Optimization},
\newblock Cambridge University Press, 2004.

\end{thebibliography}
\end{document}